\documentclass[11pt,a4paper]{article}

\usepackage{chao}
\usepackage[capitalise,noabbrev]{cleveref}
\usepackage{soul}
\usepackage{subcaption}
\usepackage{tabularx}
\usepackage{booktabs}
\usepackage{array}
\usepackage{algorithm}
\usepackage{algpseudocode}
\usepackage{listings}
\providecommand{\tightlist}{\setlength{\itemsep}{0pt}\setlength{\parskip}{0pt}}

\title{An Optimal Algorithm for Cardinality-Constrained Diameter
Partitioning}
\author{Chao
Xu\thanks{School of Computer Science and Engineering, UESTC. Email: \texttt{thechaoxu@gmail.com}.} \and Mingdong
Yang\thanks{School of Computer Science and Engineering, UESTC. Email: \texttt{1249776501@qq.com}.}}
\date{}

\providecommand{\polylog}{\operatorname{polylog}}
\providecommand{\R}{\mathbb{R}}

\makeatletter
\@ifpackageloaded{subcaption}{}{\usepackage{subcaption}}
\@ifpackageloaded{caption}{}{\usepackage{caption}}
\AtBeginDocument{%

}
\AtBeginDocument{%

}
\newcounter{pandoccrossref@subfigures@footnote@counter}
{\end{figure}%
\addtocounter{footnote}{-\value{pandoccrossref@subfigures@footnote@counter}}
\@for\f:=\global@pandoccrossref@subfigures@footnotes\do{\stepcounter{footnote}\footnotetext{\f}}%
\gdef\global@pandoccrossref@subfigures@footnotes{}}
\@ifpackageloaded{float}{}{\usepackage{float}}
\floatstyle{ruled}
\@ifundefined{c@chapter}{\newfloat{codelisting}{h}{lop}}{\newfloat{codelisting}{h}{lop}[chapter]}
\floatname{codelisting}{Listing}

\makeatother

\begin{document}

\maketitle

\begin{abstract}
Cardinality-constrained diameter partitioning asks for a partition of
\(n\) items into two classes of prescribed sizes that minimizes the
larger of the two class diameters. We give an \(O(n^2)\) algorithm and a
matching \(\Omega(n^2)\) lower bound if we can only query the weight
between two elements. The algorithm computes the optimum for every
cardinality simultaneously, improving Avis's \(O(n^2\log n)\). The
reduction is to a bottleneck 2-coloring problem on the maximum spanning
tree, solved by a standard tree DP. For a single cardinality with
Euclidean weights, we obtain a subquadratic-time algorithm in any fixed
dimension.
\end{abstract}

Let \(V\) be a set of \(n\) elements with weights
\(w:\binom{V}{2}\to\mathbb{R}\) on its unordered pairs, accessed via
unit-cost edge weight queries. Given a target \(c\in\{0,\dots,n\}\),
\textbf{diameter partitioning} asks for a partition
\(V=V_1\dot\cup V_2\) with \(|V_1|=c\) minimizing
\(\max\{\mathrm{diam}(V_1),\mathrm{diam}(V_2)\},\) where
\(\mathrm{diam}(S)=\max_{uv\in\binom{S}{2}}w(uv)\), with
\(\mathrm{diam}(S)=-\infty\) for \(|S|\le 1\). \(w\) need not be metric
or nonnegative.

Avis \cite{avis1986diameter} solved the cardinality-constrained case in
\(O(n^2\log n)\) in 1986, using a threshold-graph reduction to subset
sum. Asano et al. \cite{asano1988clustering} gave \(O(n\log n)\) for
the unconstrained Euclidean case. Monma and Suri
\cite{monma1991partitioning} handled the unconstrained graph case in
\(O(n^2)\). The operations-research literature came to an equivalent
formulation independently. Fernández et al. \cite{fernandez2013maximum}
introduced \emph{maximum dispersion} in 2013. The objective is to
maximize the minimum same-class pair weight. Tran et al.
\cite{tran2026coloring} gave an \(O(n^4)\) time algorithm for the
cardinality-constrained 2-class variant, \emph{cardinality-constrained
2-anticlustering with maximum dispersion} (2-MDCC). The two formulations
are actually equivalent: negating weights turns diameter partitioning
into maximum dispersion. Hence Avis's algorithm already gives
\(O(n^2\log n)\) for 2-MDCC. We state results for the diameter
formulation. All results in our paper transfer to 2-MDCC verbatim,
either by negation or by a simple reduction.

\subparagraph*{Our contribution}\label{our-contribution}
\addcontentsline{toc}{subparagraph}{Our contribution}

\begin{enumerate}
\def\labelenumi{\arabic{enumi}.}
\tightlist
\item
  We prove an \(O(n^2)\) algorithm that computes the optimum diameter
  for every cardinality \(c\in\{0,\dots,n\}\) \emph{simultaneously}
  (\cref{cor:explicit}). This improves Avis's \(O(n^2\log n)\) algorithm
  by shaving the log factor and producing all cardinalities at once
  instead of just one.
\item
  We show an \(\Omega(n^2)\) lower bound even for a single cardinality
  and even when \(w\) is a metric taking values in \(\{1,2\}\)
  (\cref{thm:lower-bound}).
\item
  For a single cardinality, we show a subquadratic-time algorithm when
  the input is \(n\) points in \(\R^d\) with Euclidean weights
  (\cref{cor:euclidean}).
\end{enumerate}

\section{Reduction to bottleneck tree 2-coloring}\label{sec:reduction}

For a graph \(G=(V,E)\) with edge weights \(w\), let
\(G_{>\lambda}=(V,\{uv\in E:w(uv)>\lambda\})\). A 2-coloring of \(V(G)\)
is \emph{proper} if the endpoints of every edge of \(G\) receive
different colors. It has cardinality \(c\) if one of the color class has
cardinality \(c\). Avis \cite{avis1986diameter} showed that when \(G\)
is the complete graph with the weights \(w\), the optimum diameter is at
most \(\lambda\) iff \(G_{>\lambda}\) admits a proper 2-coloring of
cardinality \(c\). From now on in this section, \(G\) denotes this
complete graph on \(V\) with edge weights \(w\). Binary search over the
\(\binom{n}{2}\) edge weights takes preprocessing time \(O(n^2\log n)\)
for sorting. We avoid this by working on a maximum spanning tree (MST)
\(T\) of \(G\). The forest \(T_{>\lambda}\) replaces \(G_{>\lambda}\).
Such MST-replacement reductions for diameter-style partitioning problems
are folklore \cite{asano1988clustering, caraballo2018balanced},
descending from the MST cycle property \cite[Chapter
23]{cormen2009introduction}.

Fix an MST \(T\) of \(G\) and a proper 2-coloring \(\chi:V\to\{0,1\}\)
of \(T\). For a 2-coloring \(\varphi:V\to\{0,1\}\), write
\[\mathrm{diam}(\varphi)=\max\{\mathrm{diam}(\varphi^{-1}(0)),\mathrm{diam}(\varphi^{-1}(1))\}\]
for the \emph{partition diameter} under \(\varphi\), and
\[M(\varphi)=\max\{w(uv):uv\in E(T),\ \varphi(u)=\varphi(v)\}\] for the
heaviest \(\varphi\)-monochromatic tree edge (with
\(\max\emptyset=-\infty\)). Note \(M(\varphi)\) sees only the \(n-1\)
tree edges, while \(\mathrm{diam}(\varphi)\) uses all \(\binom{n}{2}\)
pair weights of \(G\).

\begin{lemma}\label{lem:identity}

\(G_{>\lambda}\) admits a proper 2-coloring of cardinality \(c\) iff
\(\mathrm{diam}(\chi)\le\lambda\) and \(T_{>\lambda}\) admits a proper
2-coloring of cardinality \(c\).

\end{lemma}

\begin{proof}

We first show the structural identity that drives the lemma: for every
2-coloring \(\varphi:V\to\{0,1\}\),
\begin{equation}\protect\phantomsection\label{eq:decomp}{\mathrm{diam}(\varphi)=\max\{\mathrm{diam}(\chi),M(\varphi)\}.}\end{equation}

\emph{Upper bound for eq.~\ref{eq:decomp}.} Let \(u,v\) be a
\(\varphi\)-monochromatic pair, i.e., \(\varphi(u)=\varphi(v)\). If
\(\chi(u)=\chi(v)\), then \(w(uv)\le\mathrm{diam}(\chi)\). Otherwise
\(\chi(u)\ne\chi(v)\). Consider the \(u\)-\(v\) path in \(T\). Every
tree edge is \(\chi\)-bichromatic, so the path length is odd. The number
of \(\varphi\)-bichromatic edges on the path has the parity of
\(\varphi(u)\oplus\varphi(v)=0\), hence is even. So some tree edge
\(xy\) on the path is \(\varphi\)-monochromatic. The MST bottleneck
property \cite[Chapter 23]{cormen2009introduction} gives
\(w(xy)\ge w(uv)\). So \(w(uv)\le M(\varphi)\).

\emph{Lower bound for eq.~\ref{eq:decomp}.} Any
\(\varphi\)-monochromatic tree edge has its endpoints in the same
\(\varphi\)-class, so \(\mathrm{diam}(\varphi)\ge M(\varphi)\). Pick
\(u,v\) with \(\chi(u)=\chi(v)\) and \(w(uv)=\mathrm{diam}(\chi)\). If
\(\varphi(u)=\varphi(v)\), then \(u,v\) are \(\varphi\)-monochromatic,
giving \(\mathrm{diam}(\varphi)\ge w(uv)=\mathrm{diam}(\chi)\) directly.
Otherwise the \(u\)-\(v\) path has even length (since
\(\chi(u)=\chi(v)\)) but an odd number of \(\varphi\)-bichromatic edges
(since \(\varphi(u)\ne\varphi(v)\)), so some tree edge on it is
\(\varphi\)-monochromatic of weight at least \(\mathrm{diam}(\chi)\).

\emph{From eq.~\ref{eq:decomp} to the iff.} A coloring \(\varphi\)
properly 2-colors \(G_{>\lambda}\) iff every \(\varphi\)-monochromatic
pair in \(G\) has weight \(\le\lambda\), i.e.,
\(\mathrm{diam}(\varphi)\le\lambda\). Similarly \(\varphi\) properly
2-colors \(T_{>\lambda}\) iff \(M(\varphi)\le\lambda\). By
eq.~\ref{eq:decomp}, \(\mathrm{diam}(\varphi)\le\lambda\) iff
\(\mathrm{diam}(\chi)\le\lambda\) and \(M(\varphi)\le\lambda\). The
condition \(\mathrm{diam}(\chi)\le\lambda\) does not depend on
\(\varphi\), so quantifying over cardinality-\(c\) colorings \(\varphi\)
gives the iff statement.

\end{proof}

So minimizing \(\mathrm{diam}(\varphi)\) at cardinality \(c\) reduces to
two pieces: the constant \(\mathrm{diam}(\chi)\), and the smallest
\(\lambda\) for which \(T_{>\lambda}\) admits a proper 2-coloring of
cardinality \(c\). Define
\[M^*(c)=\min\{M(\varphi):\varphi:V\to\{0,1\},\ |\varphi^{-1}(0)|=c\},\]
the minimum heaviest \(\varphi\)-monochromatic tree edge over
cardinality-\(c\) colorings. Equivalently, \(M^*(c)\) is the smallest
tree-edge weight \(\lambda\) such that \(T_{>\lambda}\) admits a proper
2-coloring of cardinality \(c\), or \(-\infty\) if \(T\) itself does.

By \cref{lem:identity}, the optimum diameter at cardinality \(c\) equals
\(\max\{\mathrm{diam}(\chi),M^*(c)\}\). Computing \(T\), \(\chi\), and
\(\mathrm{diam}(\chi)\) takes \(O(n^2)\) time by Prim
\cite{prim1957shortest} and brute-force diameter. Hence it only remains
to compute \(M^*(c)\). This leads us to the \emph{bottleneck tree
2-coloring} problem: Given a weighted tree \(T\) and a cardinality
\(c\in\{0,\dots,|V(T)|\}\), compute \(M^*(c)\).

\section{Algorithms for bottleneck tree
2-coloring}\label{sec:algorithms}

\begin{theorem}\label{thm:single-c}

Bottleneck tree 2-coloring for a single cardinality can be computed in
\(O(n\polylog n)\) time.

\end{theorem}

\begin{proof}

Recall \(M^*(c)\) is the smallest tree-edge weight \(\lambda\) such that
\(T_{>\lambda}\) admits a proper 2-coloring of cardinality \(c\), or
\(-\infty\) if \(T\) itself does. Test whether \(\chi\) (or its swap)
has cardinality \(c\) to decide the \(-\infty\) case. Otherwise sort the
\(n-1\) tree-edge weights and binary-search with \(O(\log n)\)
feasibility tests on \(T_{>\lambda}\). Each test reduces to a subset-sum
instance of non-negative integers with total sum \(O(n)\)
\cite{avis1986diameter}. Koiliaris-Xu \cite[Theorem
1.1]{koiliaris2019faster} decides such instances in \(O(n\polylog n)\)
time.

\end{proof}

\begin{theorem}\label{thm:tree-dp}

Bottleneck tree 2-coloring for all cardinalities can be computed in
\(O(n^2)\) time. An optimal labeling for any chosen cardinality can be
reconstructed in \(O(n)\) additional time.

\end{theorem}

\begin{proof}

We design a dynamic programming algorithm. Root \(T\) at any vertex
\(r\). For each node \(u\) with rooted subtree \(T_u\), let \(D_u^b[q]\)
be the smallest \(M(\varphi)\) over 2-colorings
\(\varphi:V(T_u)\to\{0,1\}\) with \(\varphi(u)=b\) and
\(|\varphi^{-1}(0)\cap V(T_u)|=q\), where \(M(\varphi)\) uses only edges
of \(T_u\) (and \(D_u^b[q]=+\infty\) if no such \(\varphi\) exists).

For a leaf \(u\), \(D_u^0[1]=D_u^1[0]=-\infty\) and other entries are
\(+\infty\). For an internal node \(u\) with children \(v_1,\dots,v_k\):
\[D_u^b[q] = \min_{\substack{g_i\in\{0,1\}\\ q_1+\cdots+q_k = q-[b=0]}}\;\max_{i\in[k]}\max\big\{D_{v_i}^{g_i}[q_i],\ [b=g_i]\cdot w(uv_i)\big\},\]
where \([\cdot]\) is the Iverson bracket: \([P]=1\) if \(P\) is true and
\(0\) otherwise. The product \([P]\cdot w\) with an edge weight uses the
max-plus convention: \(w\) if \(P\) is true and \(-\infty\) otherwise.
Each child \(v_i\) contributes its subtree optimum at the chosen color
\(g_i\), plus the edge \(uv_i\) if it becomes \(\varphi\)-monochromatic.

To compute the recurrence efficiently, \emph{absorb} children into a
\emph{running table} \(R^b[q]\) one at a time. Initialize \(R\) to
represent \(u\) alone (the leaf base case applied to \(R\)). For each
child \(v\) in turn, replace \(R\) with
\[R'^b[Q]=\min_{\substack{q_1+q_2=Q\\ g\in\{0,1\}}}\max\big\{R^b[q_1],\ D_v^g[q_2],\ [b=g]\cdot w(uv)\big\}.\]
After all \(k\) children of \(u\) are absorbed, \(D_u\gets R\).

If the current \(R\) represents \(s\) vertices, absorbing \(v\) visits
\(s\cdot|T_v|\) pairs \((q_1,q_2)\) in \(O(1)\) time each, costing
\(O(s\cdot|T_v|)\).

For total time, each absorption cost \(s\cdot|T_v|\) counts the
unordered vertex pairs \(\{x,y\}\) with \(x\) in \(R\) and
\(y\in V(T_v)\). Every unordered pair \(\{x,y\}\subseteq V\) is counted
at exactly one absorption: at \(u=\mathrm{LCA}(x,y)\), exactly one of
\(x,y\) is in \(R\) before \(u\) absorbs the child whose subtree
contains the other. So total cost \(=\binom{n}{2}=O(n^2)\).

At the root, \(M^*(c)=\min(D_r^0[c],D_r^1[c])\). Storing argmin choices
at each absorption lets us walk back from \(r\) in \(O(n)\) time and
recover an optimal \(\varphi\).

\end{proof}

Combining \cref{lem:identity} with \cref{thm:tree-dp} gives the optimum
diameter for every cardinality \(c\) in \(O(n^2)\) time.

\begin{corollary}\label{cor:explicit}

Diameter partitioning can be computed for all cardinalities in
\(O(n^2)\) time. An optimal partition for any chosen cardinality can be
reconstructed in \(O(n)\) additional time.

\end{corollary}

We also get the same corollary for 2-MDCC.

\section{Euclidean metric}\label{sec:euclidean}

By \cref{lem:identity}, diameter partitioning reduces to one MST, one
diameter computation per color class, and one bottleneck tree 2-coloring
instance. The first two admit subquadratic algorithms in the Euclidean
metric. In the plane, the Euclidean diameter and the Euclidean MST each
admit \(O(n\log n)\) algorithms
\cite{preparata1985computational, monma1990computing}. For fixed
\(d\ge 3\), both admit \(O(n^{2-2/(\lceil d/2\rceil+1)+\epsilon})\)-time
algorithms for every \(\epsilon>0\) \cite{agarwal1992farthest}.

For 2-MDCC, the negation argument fails since negated Euclidean
distances do not form a Euclidean metric. \cref{lem:identity} still
applies after swapping max with min throughout: replace MAX-MST with
MIN-MST and the diameter computation with closest pair. In the plane,
the Euclidean MST and the Euclidean closest pair both admit
\(O(n\log n)\) algorithms \cite{preparata1985computational}. For fixed
\(d\ge 3\), the Euclidean MST admits
\(O(n^{2-2/(\lceil d/2\rceil+1)+\epsilon})\) for every \(\epsilon>0\)
\cite{agarwal1991euclidean}. Closest pair stays at \(O(n\log n)\) in
any fixed dimension \cite{preparata1985computational}, but the MST step
dominates.

Combined with \cref{thm:single-c}:

\begin{corollary}\label{cor:euclidean}

For \(n\) points in \(\R^d\) under the Euclidean metric, both diameter
partitioning and 2-MDCC for a single cardinality can be solved in
\(O(n\polylog n)\) time when \(d=2\). For fixed \(d\ge 3\), both can be
solved in \(O(n^{2-2/(\lceil d/2\rceil+1)+\epsilon})\) time for every
\(\epsilon>0\).

\end{corollary}

\section{Lower bound}\label{sec:lower-bound}

\begin{theorem}\label{thm:lower-bound}

For infinitely many \(n\), any randomized algorithm that outputs the
exact optimum of diameter partitioning or 2-MDCC on \(n\) vertices with
success probability at least \(2/3\) makes \(\Omega(n^2)\) edge weight
queries in the worst case, even when the weights form a metric taking
only the values \(1\) and \(2\).

\end{theorem}

\begin{proof}

Take \(n\ge 4\) even and set \(c=n/2\). Partition \(V\) into known sets
\(A\) and \(B\) with \(|A|=|B|=n/2\). Set \(w(ab)=2\) for every
\(a\in A,\ b\in B\), and \(w(uv)=1\) for every \(\{u,v\}\subseteq B\).
For each pair \(\{u,v\}\subseteq A\), the adversary sets
\(w(uv)\in\{1,2\}\).

These instances are metric. Weights lie in \(\{1,2\}\), so
\(w(xz)\le 2\le w(xy)+w(yz)\) for any triple \(x,y,z\in V\).

By construction \(\mathrm{diam}(B)=1\). If every pair within \(A\) also
has weight \(1\), then \(\mathrm{diam}(A)=1\), so the partition
\((A,B)\) achieves \(\max\{\mathrm{diam}(A),\mathrm{diam}(B)\}=1\). The
optimum is \(1\). If some pair \(\{a,a'\}\subseteq A\) has weight \(2\),
the optimum is \(2\). Take any partition with \(|V_1|=n/2\). Since
\(|A|=|B|=n/2\), either \(V_1\in\{A,B\}\) or \(V_1\) contains both an
element of \(A\) and an element of \(B\). In the first case \(a,a'\)
both lie in class \(A\), so \(\mathrm{diam}(A)\ge w(aa')=2\). In the
second case \(V_1\) contains some \(a''\in A\) and \(b\in B\), so
\(\mathrm{diam}(V_1)\ge w(a''b)=2\). Either way the partition diameter
is at least \(2\). Weights are bounded by \(2\), so the optimum is
\(2\).

Hence distinguishing optimum \(1\) from optimum \(2\) amounts to
deciding whether some pair within \(A\) has weight \(2\). This is the OR
problem on the \(\binom{n/2}{2}=\Theta(n^2)\) bits indexed by
\(\binom{A}{2}\). Its bounded-error randomized query complexity is
\(\Theta(n^2)\) \cite{buhrman2002complexity}. Each edge weight query
reveals one bit, so the algorithm makes \(\Omega(n^2)\) edge weight
queries.

For 2-MDCC, mirror the construction: set \(w(ab)=1\) for
\(a\in A,\ b\in B\), \(w(uv)=2\) for \(\{u,v\}\subseteq B\), and
adversarial \(w(uv)\in\{1,2\}\) within \(A\). This is still a
\(\{1,2\}\)-metric. On the partition \((A,B)\), every pair within \(B\)
has weight \(2\). So the 2-MDCC objective equals the minimum pair weight
within \(A\). Any other partition with \(|V_1|=n/2\) places both an
element of \(A\) and an element of \(B\) in some class. That class
contains a pair of weight \(1\), so the objective is at most \(1\). Thus
the optimum is \(2\) iff every pair within \(A\) has weight \(2\),
otherwise \(1\). The same OR reduction gives \(\Omega(n^2)\) edge weight
queries.

\end{proof}

\bibliographystyle{plainurl}
\bibliography{references.bib}

\end{document}